\documentclass[10pt,usletter]{article}
\usepackage[utf8]{inputenc}
\usepackage[T1]{fontenc}
\usepackage{amsthm,amssymb,amsmath,amsfonts}
\usepackage[in]{fullpage}
\usepackage{thm-restate}
\usepackage{hyperref}
\usepackage{cleveref}
\usepackage{setspace}
\usepackage{orcidlink}
\usepackage{enumitem}

\usepackage{comment}
\usepackage{scalerel}
\usepackage[all]{nowidow}

\title{Computing Ultrabubbles in Linear Time by Reduction to Weak Superbubbles in the Double Covering Graph} %
\date{}

\newtheorem{theorem}{Theorem}
\newtheorem{observation}[theorem]{Observation}
\newtheorem{definition}[theorem]{Definition}
\newtheorem{lemma}[theorem]{Lemma}

\crefformat{footnote}{#2\footnotemark[#1]#3}

\author{Sebastian Schmidt \and
Juha Harviainen \and
Corentin Moumard \and
Aleksandr Politov \and
Francisco Sena \and
Alexandru I. Tomescu
}

\newcommand{\signs}{\ensuremath{\{+,-\}}}

\begin{document}

{
\Huge
\centering
Computing Ultrabubbles in Linear Time by Reduction to Weak Superbubbles in the Double Covering Graph
}

\vspace*{1em}

{
\Large
\centering
\noindent
Sebastian~Schmidt\footnote{\label{foot:address}\texttt{firstname.lastname@helsinki.fi}, Department of Computer Science, University of Helsinki, Finland}\orcidlink{0000-0003-4878-2809},
Juha~Harviainen\cref{foot:address}\footnote{\label{foot:contribution}Equal contribution}\orcidlink{0000-0002-4581-840X},
Corentin~Moumard\cref{foot:address}\cref{foot:contribution},\newline
Aleksandr~Politov\cref{foot:address}\cref{foot:contribution},
Francisco~Sena\cref{foot:address}\cref{foot:contribution}\orcidlink{0000-0002-3508-4473} and
Alexandru~I.~Tomescu\cref{foot:address}\footnote{Corresponding author}\orcidlink{0000-0002-5747-8350}
}

\begin{abstract}
In recent years, in attempts to capture variation structures in genomic and pangenomic graphs, the zoo of \emph{bubbles} has been filled with more and more different variants.
The first more complex variant was the \emph{superbubble} which was later (informally) generalised to bidirected graphs as the \emph{ultrabubble}.
Ultrabubbles are more applicable in bioinformatics research, because genome and pangenome graphs usually are bidirected.

In this paper, we study the relationship between superbubbles and ultrabubbles, and show that ultrabubbles in a bidirected graph are equivalent to the weak superbubbles in the corresponding double covering graph.
Using this, we introduce a new linear-time algorithm to enumerate all ultrabubbles that consists of a very simple reduction to employ an existing linear-time algorithm to enumerate all weak superbubbles.
Treating the superbubble algorithm as a black box, our algorithm is vastly simpler than the SPQR-tree-based linear-time ultrabubble enumeration algorithm of Sena et al. [2026].
Further, since our algorithm works on all bidirected graphs, it is strictly more powerful than the simple graph-orientation-based algorithm by Harviainen et al. [2026], which only works on a restricted class of bidirected graphs.

In order to prove our equivalence, we additionally introduce simplified definitions of both ultrabubbles and weak superbubbles, through which the equivalence becomes more intuitive.
Finally, our reduction motivates that some algorithmic problems in bidirected graphs can be reduced to well-known algorithmic problems in directed graphs.
\end{abstract}

\paragraph*{Funding.}
Co-funded by the European Union (ERC, SCALEBIO, 101169716). Views and opinions expressed are however those of the author(s) only and do not necessarily reflect those of the European Union or the European Research Council. Neither the European Union nor the granting authority can be held responsible for them. Co-funded also by the Research Council of Finland, Grants 346968, 358744 and 374907.
\scalerel*{\includegraphics{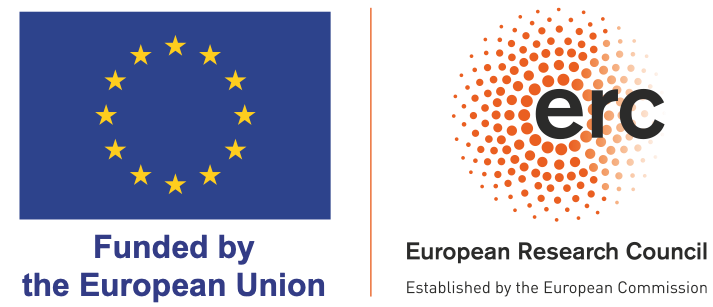}}{Gg}

\paragraph*{Acknowledgements.}
We thank the anonymous reviewers for their helpful thorough comments on previous versions of this paper.

\pagebreak

\section{Introduction}

\subsection{Bidirected graphs}

\begin{figure}
    \centering
    \includegraphics[scale=1]{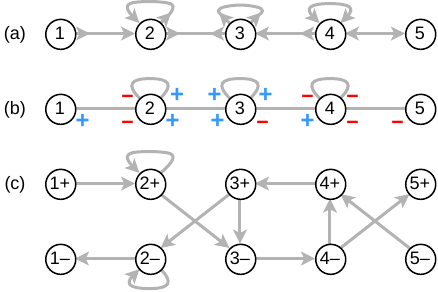}
    \caption{Three representations of the same bidirected graph.
    \textbf{(a)} Arrow representation.
    An arc can be incoming or outgoing at both ends.
    If a walk enters a node incoming, it must leave outgoing, and if a walk enters outgoing, it must leave incoming.
    Arcs with two outgoing ends allow a walk to switch from forward to reverse, and two incoming ends allow a walk to switch from reverse to forward.
    The sequence $1, 2, 2, 3, 4, 4, 3, 3, 4, 5$ is a walk $W$.
    \textbf{(b)} Signed representation.
    A walk alternates signs at each node.
    $W$ is represented by the sequence $1+, 2+, 2+, 3-, 4-, 4+, 3+, 3-, 4-, 5+$.
    \textbf{(c)} Double covering graph.
    Multiarcs caused by the self loops have been collapsed into a single arc.
    This is a directed graph produced by the doubling operation applied to (b).
    Directed self loops become directed self loops in the double covering graph, while $++$ and $--$ self loops become arcs between $-$ and $+$ of the same node.
    $W$ is represented by the sequence $1+, 2+, 2+, 3-, 4-, 4+, 3+, 3-, 4-, 5+$.}
    \label{fig:doubling}
\end{figure}

\emph{Bidirected graphs} (also known as \emph{polarized graphs}~\cite{zelinka1974polarizedgraphs}) are a generalisation of directed graphs where arcs can also be incoming to both their incident nodes, or outgoing from both their incident nodes.
In mathematical contexts, bidirected graphs are usually depicted as undirected graphs with plus or minus \emph{signs} at the incidences.
A walk must obey the \emph{sign-alternation} rule, meaning that when it enters a node with a $-$ ($+$), then it must leave it through a $+$ ($-$).
The edges in a bidirected graph are represented by unordered pairs of \emph{signed nodes} $v\alpha$ which are pairs of a node $v$ and a sign $\alpha \in \{+, -\}$.
See \Cref{fig:doubling}~(b) for an example of this definition and see \Cref{fig:doubling}~(a) for a visualisation that draws the connection to directed graphs.
Bidirected graphs are heavily used in practice in bioinformatics~\cite{bessouf2019transitivebidirected,medvedev2007computability,kita2017bidirectedgraphsisigned}, besides being also interesting generalisation of directed graphs from a theoretical perspective~\cite{schrijver2003combinatorial}.

While bidirected graphs are a simple generalisation of directed graphs, graph algorithms for directed graphs require at least extra consideration when applied to bidirected graphs, and at worst completely different approaches.
Typically, graph algorithms have to be modified to track the signs of nodes.
For example, Dijkstra's classic shortest path algorithm~\cite{dijkstra1959note} can be applied to bidirected graphs by representing the two signed variants of nodes as two separate nodes in the Dijkstra data structures, and specifying the start and target nodes with signs.
In the same way, reachability-based algorithms can be modified.
As these modifications are rather straightforward, they are quietly being used in practice when applicable~\cite{chang2025giraffe,dabbaghie2022bubblegun,garg2018graph,iqbal2012novo,shafin2020nanopore,matchtigs}.
However, not all cases are that simple, and in this work we consider one such case.

\subsection{Bubbles}

\begin{figure}
    \centering
    \includegraphics[scale=1]{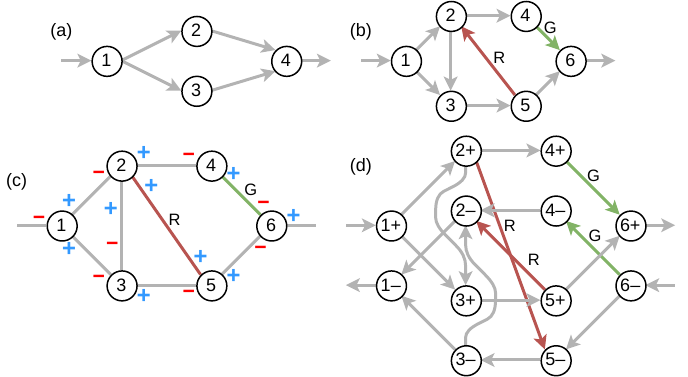}
    \caption{\textbf{(a)} The classic notion of a bubble: a path that diverges and then merges. 
    \textbf{(b)} Due to the cycle $2, 3, 5, 2$, the pair $(1, 6)$ is not a superbubble.
    If the red arc R is removed, then the cycle is broken and $(1, 6)$ is a superbubble.
    If additionally the green arc G is removed, then $4$ is a tip, and hence $(1, 6)$ is not a superbubble.
    \textbf{(c)} Same graph as in (b), but the red edge R  is not directed because its endpoints have the same signs.
    Due to the cycloid $2+, 3+, 5+, 2-$, the pair $\{1+, 6-\}$ is not an ultrabubble.
    If the red edge R is removed, then the cycloid is broken and $\{1+, 6-\}$ is an ultrabubble.
    If additionally the green edge G is removed, then $4$ is a tip, and hence $\{1+, 6-\}$ is not an ultrabubble.
    \textbf{(d)} The double covering graph of the graph (c).
    Node $5-$ is reachable from $1+$ but not reverse-reachable from $6+$, hence $(1+, 6+)$ is not a superbubble.
    Symmetrically, node $5+$ is reverse reachable from $1-$ but not reachable from $6-$, hence $(6-, 1-)$ is not a superbubble.
    If the red arcs R are removed, then both $(1+, 6+)$ and $(6-, 1-)$ are superbubbles.
    If additionally the green arcs G are removed, then $4+$ and $4-$ are tips, and hence neither $(1+, 6+)$ nor $(6-, 1-)$ is a superbubble.}
    \label{fig:bubbles}
\end{figure}

A \emph{bubble} is a biologically significant region of a graph, and there are various definitions made to capture various biological interpretations of graph regions~\cite{spades,paten2018ultrabubbles,billi,li2024exploring}.
See \Cref{fig:bubbles}~(a) for an example.
Common among the various definitions is typically a single entrance and a single exit node, and the idea that the bubble consists of two or more separate paths from entrance to exit.
Some definitions of bubbles have also been shown to be useful as a clustering method of biological graphs.
For example, a method like \emph{multilevel Dijkstra}~\cite{delling2011customizablerouteplanning} can be used to speed up shortest path queries based on \emph{snarl}s~\cite{chang2020snarlshortestpaths} in biological graphs.
Multilevel Dijkstra precomputes shortcuts that skip a bubble in order to reduce the size of the search space for the shortest path.
We expect that the same technique could also be used to speed up \emph{maximum flow} computations in practice.

In this paper we focus on \emph{superbubbles}~\cite{onodera2013detecting} and \emph{ultrabubbles}~\cite{paten2018ultrabubbles} which capture for example variation structures in pangenome graphs~\cite{harviainen2026orientation}.

A \emph{superbubble} $(u, v)$ is an acyclic subgraph of a directed graph that can be entered only via its only source $u$ and exited only via its only sink $v$.
See \Cref{fig:bubbles}~(b) for an example.
The superbubbles of a graph can be enumerated in time $O(n(m + n))$ with the algorithm of Onodera~et~al.~\cite{onodera2013detecting}, where $n$ is the number of nodes and $m$ the number of arcs.
This was later reduced to linear time first for only directed acyclic graphs~\cite{brankovic2016linearsuperbubblesindags} and then for general directed graphs~\cite{gärtner-revisited} and the latter was later simplified by Gärtner~et~al.~\cite{gärtner2019direct}.
The practical relevance of superbubbles is underlined by the fact that they carry biological meaning and occur in the order of millions in biological graphs~\cite{harviainen2026orientation}.
Additionally, superbubbles have also been studied outside of bioinformatics as a tool for network analysis~\cite{gärtner2021superbubblesnetworkscience}.
A variant of superbubbles are \emph{weak} superbubbles, which additionally allow the arc $(v, u)$ to exist~\cite{gärtner-revisited}.
See \Cref{fig:weak-superbubbles} for an example.

As a generalisation of superbubbles to bidirected graphs, \emph{ultrabubbles} were defined by Paten~et~al.~\cite{paten2018ultrabubbles}.
They are defined like superbubbles, as a pair of signed nodes $\{u\alpha, v\beta\}$ that define an acyclic subgraph that is only connected to the rest of the graph via signed nodes $u\hat{\alpha}$ and $v\hat{\beta}$, where $\hat{+} = -$ and $\hat{-} = +$ is the \emph{opposite} of a sign.
Ultrabubbles are also a useful clustering method for speeding up shortest-path and maximum flow computations, and also occur in the order of millions in biological graphs~\cite{harviainen2026orientation}.
They are actually more applicable in bioinformatics than superbubbles, as sequence graphs in bioinformatics are typically bidirected, due to the reverse complementarity of DNA.

Even though ultrabubbles are a generalisation of superbubbles to bidirected graphs, initial theoretical attempts at enumeration algorithms did not make use of any superbubble algorithm.
Paten~et~al. proposed a simple $O(n(m + n))$-time enumeration algorithm, and Zisis and Sætrom~\cite{zisis2026ultrabubbles} improved that to $O(Kn+n+m)$ time, where $K$ is the number of ``snarls'' in the graph.
A snarl is a structure similar, but more general, than ultrabubbles, that may contain e.g.~cycles.
The first linear-time algorithm enumerating all ultrabubbles of a bidirected graph was proposed by Sena~et~al.~\cite{sena2026spqr-bubbles}.
It is a two-step approach that first computes the SPQR tree as a representation of all 2-cuts on the undirected graph and then identifies which 2-cuts are ultrabubbles.
While the algorithm of Sena~et~al. can also compute superbubbles, it does not draw a connection between superbubbles and ultrabubbles other than that both are based on 2-cuts.
Besides that, it uses two separate bodies of theory resulting in two separate, although similar, algorithms for computing ultrabubbles and superbubbles.

\subsection{Double Covering Graphs}

\emph{Double covering graphs} fall into the class of \emph{skew-symmetric graphs} as defined by Goldberg and Karzanov~\cite{goldberg1996skewsymmetricgraphs}.
Skew-symmetric graphs were first introduced as \emph{antisymmetrical digraphs} by Tutte~\cite{tutte1967antisymmetricaldigraphs}, and later also called \emph{double covering graphs of polar graphs}~\cite{zelinka1976doublecoveringpolar} and \emph{double covering graphs of bidirected graphs}~\cite{zaslavsky1991doublecoveringbidirected}.
We call the process of forming the double covering graph of a bidirected graph the \emph{doubling operation}.
It can be formalised as follows~\cite{zelinka1976doublecoveringpolar,zaslavsky1991doublecoveringbidirected}.
Given a bidirected graph, each node $v$ is replaced by a pair of nodes $v+$ and $v-$.
Each edge $\{u\alpha, v\beta\}$ is replaced by a pair of arcs $(u\alpha, v\hat{\beta})$ and $(v\beta, u\hat{\alpha})$.
For example, a bidirected edge $\{a+, b+\}$ would be replaced by two arcs $(a+, b-)$ and $(b+, a-)$.
See \Cref{fig:doubling}~(b) and~(c) for an example of this operation.

Note that the doubling operation maintains connectivity and preserves walks between signed nodes (see \Cref{lem:connectivity} below).
Intuitively, it is also easy to see that for example shortest paths are maintained by doubling if the weights of the bidirected edges are copied to the corresponding inserted arcs.
While not proven formally, this insight was used in practical research already~\cite{matchtigs}.
Also, while not explicitly mentioned, the tool BubbleGun~\cite{dabbaghie2022bubblegun} computes a structure similar to ultrabubbles in a bidirected graph by implicitly doubling the bidirected graph and computing weak superbubbles in the resulting (directed) double covering graph.
Harviainen~et~al.~\cite{harviainen2026orientation} later noticed that the output of BubbleGun is almost identical to ultrabubbles, as empirical research has shown that BubbleGun's output differs only in a few cases per graph on practical instances.

Given that graph doubling is a well-known, yet almost always implicitly used technique to apply directed graph algorithms to bidirected graphs, this raises the question if this technique can also be used to enumerate ultrabubbles.

\begin{figure}
    \centering
    \includegraphics[scale=1]{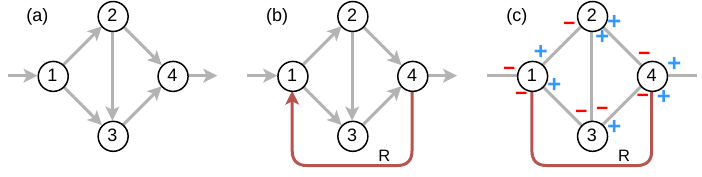}
    \caption{\textbf{(a)} $(1, 4)$ is both a superbubble and a weak superbubble.
    \textbf{(b)} Due to the red arc R, $(1, 4)$ is a weak superbubble but not a superbubble.
    \textbf{(c)} Writing (b) as a bidirected graph, we see that $\{1+, 4-\}$ is an ultrabubble.}
    \label{fig:weak-superbubbles}
\end{figure}

To achieve this, we draw inspiration from another attempt to make use of the fact that ultrabubbles are a generalisation of weak superbubbles.
Harviainen~et~al.~\cite{harviainen2026orientation} propose a \emph{sign-flipping}-based orientation algorithm to compute ultrabubbles in a restricted class of bidirected graphs.
The sign-flipping operation flips the signs of all incident edges of a single node.
This has the effect that all reachabilities are maintained and no new ones are created.
Harviainen~et~al. noticed that in a bidirected graph where all edges are of the form $+-$ or $-+$ (i.e.~only directed edges exist), the ultrabubbles are equivalent to weak superbubbles as opposed to superbubbles, which ultrabubbles are supposed to generalise.

In their algorithm, Harviainen~et~al.~\cite{harviainen2026orientation} attempt to orient the bidirected graph into a directed graph by sign-flipping.
However, it is only possible to orient a bidirected graph in this way if it contains no cycle with an odd number of non-directed edges (i.e.~edges of the form $++$ or $--$)~\cite{harary1953notion,harviainen2026orientation}.
Harviainen~et~al. handle such non-orientable graphs by applying their orientation algorithm and inserting a source or sink node into any edge that the algorithm is unable to orient.
Then they apply the algorithm of Gärtner~et~al.~\cite{gärtner2019direct} to compute weak superbubbles in the oriented graph.
They show that their algorithm works on bidirected graphs that contain a cut node or a source or sink (\emph{tip}).

Besides having the drawback of working only for a restricted class of bidirected graphs, the algorithm by Harviainen~et~al.~\cite{harviainen2026orientation} does not really draw a connection between directed and bidirected graphs, as the orientation approach does not maintain reachabilities.
Further, due to the handling of conflicts being catered to finding ultrabubbles, it is unclear how this approach generalises to other algorithmic problems.
On the other hand, the doubling approach is already being used successfully in simpler practical cases in order to compute for example shortest paths in bidirected graphs~\cite{matchtigs,chang2020snarlshortestpaths}, and hence there is evidence that it generalises better to other algorithmic problems.
In this paper, we show that via graph doubling, it is simple to enumerate ultrabubbles in bidirected graphs in linear time, when an algorithm to enumerate weak superbubbles in directed graphs in linear time is given.
See \Cref{fig:bubbles}~(c) and~(d) for an example of our technique.
This provides evidence that graph doubling is not just a technique for simple cases, but can also be applied when the reduction from bidirected to directed is not immediately obvious.

\subsection{Main contributions}

\begin{itemize}
    \item
    We introduce a simple reduction from ultrabubbles to weak superbubbles via the graph doubling.
    This \textbf{connects existing theoretical work on ultrabubbles and superbubbles} and motivates that graph doubling could also draw connections between bidirected and directed graphs in other cases.
    
    \item 
    We show how to enumerate ultrabubbles by reducing to enumerating weak superbubbles in the double covering graph.
    This results in the \textbf{first linear-time reduction-based algorithm for computing ultrabubbles} without further structural assumptions about the input graph.

    \item
    We transform the definitions of ultrabubbles and weak superbubbles into a short set of simple conditions that are equivalent modulo bidirectedness.
    This makes our reduction via doubling especially simple, and motivates that \textbf{our simplified definitions may be of independent interest for future algorithms} involving ultrabubbles or weak superbubbles.
\end{itemize}

\subsection{Organisation of the paper}

For our reduction, we first simplify and harmonise the definitions of both ultrabubbles (\Cref{s:ultrabubbles}) and superbubbles (\Cref{s:superbubbles}), such that their equivalence becomes more obvious.
Then, in \Cref{s:reduction}, we describe the reduction.
Omitted proofs are in Appendix~\ref{apx:proofs}.

\section{Preliminaries}

We use standard definitions of directed and bidirected graphs, and use the same terms for bidirected and directed definitions.

\textbf{Directed graphs.}
A \emph{directed graph} $G = (V, E)$ has a set of nodes $V = V(G)$ and a set of (directed) arcs $E = E(G)$.
An \emph{arc} is a tuple $(u, v)$, where $u$ and $v$ are nodes.
The \emph{induced subgraph} $G[U]$ in $G$ of the set of nodes $U \subseteq V(G)$ is the graph $G[U] := (U, \{(u, v) \in E(G) \mid u, v \in U\})$.
The \emph{reverse graph} of $G$ is $\overleftarrow{G} := (V(G), \{(u, v) \mid (v, u) \in E(G)\})$.
A \emph{$v_1$-$v_\ell$-walk} in a directed graph is a sequence of nodes $v_1, \dots, v_\ell$ such that $(v_i, v_{i+1}) \in E$.
A \emph{cycle} is a $v$-$v$-walk for any $v \in V$ that contains at least one arc.
An \emph{open path} (or just \emph{path}) is a walk that repeats no node.
A \emph{closed path} is a cycle that repeats no node, except that the start and end are the same node.
A \emph{tip} is a node with only incoming or only outgoing arcs.

\textbf{Bidirected graphs.}
A \emph{bidirected graph} $G = (V, E)$ has a set of nodes $V = V(G)$ and a set of bidirected edges $E = E(G)$.
A \emph{sign} is a symbol $\alpha \in \signs$.
The \emph{opposite sign} $\hat{\alpha}$ of $\alpha$ is defined as $\hat{+} = -$ and $\hat{-} = +$.
A \emph{signed node} is a pair $(v, \alpha)$ with $v \in V$ and $\alpha \in \signs$.
We concisely write a signed node as $v\alpha$.
A \emph{bidirected edge} (or simply just \emph{edge}) is an unordered pair of signed nodes $\{u\alpha, v\beta\}$.
A \emph{$v_1\alpha_1$-$v_\ell\alpha_\ell$-walk} is a sequence of signed nodes $v_1\alpha_1, \dots, v_\ell\alpha_\ell$ such that $\{v_i\alpha_i, v_{i+1}\hat{\alpha}_{i+1}\} \in E$.
(Note that our definition of bidirected walks is written to resemble directed walks, and to avoid special cases we use the convention that the sign of nodes in a walk is always the sign with which the walk leaves the node.
Even if the walk never leaves the last node $v_\ell$, we say that it ends with the leaving sign $\alpha_\ell$.
Typically, definitions of bidirected walks handle the last node the other way around, and say that a walk ends with the incoming sign $\hat{\alpha_\ell}$.
See \Cref{fig:doubling}~(b) for an example.)
A \emph{cycloid} is a $v\alpha$-$v\beta$-walk for any $v \in V$ that contains at least one edge and where $\alpha$ and $\beta$ may or may not be equal.
A \emph{cycle} is a cycloid with $\alpha = \beta$.
A \emph{path} is a walk that repeats no node (i.e. for each node $v$, it may contain at most one of $v+$ and $v-$).
An \emph{undirected $v_1\alpha_1$-$v_\ell\alpha_\ell$-walk} in a bidirected graph is a sequence of nodes $v_1, \dots, v_\ell$ such that for each $i$ there exist some $\alpha'_i, \beta'_{i+1} \in \signs$ such that $\{v_i\alpha'_i, v_{i+1}\beta'_{i+1}\} \in E$ and such that $\alpha_1 = \alpha'_1$ and $\beta_\ell = \beta'_\ell$.
An \emph{undirected path} is an undirected walk that repeats no node.
A \emph{tip} is a node with edges incident to only one of its signed nodes.

\textbf{The doubling operation.}
The \emph{doubling operation} transforms a bidirected graph $G = (V, E)$ into a the corresponding directed double covering graph $G' = (V', E')$.
For each node $v \in V$, we add two signed nodes $v+, v-$ into $V'$.
For each edge $\{u\alpha, v\beta\} \in E$, we add two arcs $(u\alpha, v\hat{\beta})$ and $(v\beta, u\hat{\alpha})$ into $E'$.
See \Cref{fig:doubling}~(b) and~(c) for an example.
The following lemma proves that the doubling operation maintains connectivity and preserves walks.

\begin{restatable}{lemma}{lemconnectivity}
    \label{lem:connectivity}
    Let $G$ be a bidirected graph and $G'$ the corresponding double covering graph.
    Then $v_1\alpha_1, \dots, v_\ell\alpha_\ell$ is a walk in $G$ if and only if $v_1\alpha_1, \dots, v_\ell\alpha_\ell$ is a walk in $G'$.
\end{restatable}
\begin{proof}
    $(\Rightarrow)$
    Let $v_1\alpha_1, \dots, v_\ell\alpha_\ell$ be a walk in $G$.
    Then by the definition of the doubling operation, $G'$ contains all nodes $v_1\alpha_1, \dots, v_\ell\alpha_\ell$.
    And since the edges $\{v_i\alpha_i, v_{i+1}\hat{\alpha}_{i+1}\}$ exist in $G$, the arcs $(v_i\alpha_i, v_{i+1}\alpha_{i+1})$ exist in $G'$.
    Hence, $v_1\alpha_1, \dots, v_\ell\alpha_\ell$ is a walk in $G'$.

    $(\Leftarrow)$
    Let $v_1\alpha_1, \dots, v_\ell\alpha_\ell$ be a walk in $G'$.
    Then by the definition of the doubling operation, $G$ contains all nodes $v_1, \dots, v_\ell$.
    And since the arcs $(v_i\alpha_i, v_{i+1}\alpha_{i+1})$ exist in $G'$, the edges $\{v_i\alpha_i, v_{i+1}\hat{\alpha}_{i+1}\}$ exist in $G$.
    Hence, $v_1\alpha_1, \dots, v_\ell\alpha_\ell$ is a walk in $G$.
\end{proof}

\begin{observation}
    \label{obs:doubling}
    The doubling operation can be performed in $O(n + m)$ time.
\end{observation}

Note that in a double covering graph, each walk has a corresponding reversed walk.

\begin{definition}[Walk opposite in double covering graphs]
    \label{def:opposite-walk}
    Let $G$ be a bidirected graph and $G'$ the corresponding double covering graph.
    Let $W' := v_1 \alpha_1, \dots, v_\ell \alpha_\ell$ be a walk in $G'$.
    Then the \emph{opposite} $\hat{W}' := v_\ell \hat{\alpha}_\ell, \dots, v_1 \hat{\alpha}_1$ of $W'$ is formed by reversing $W'$ while taking the opposite of all signs.
\end{definition}

\begin{lemma}
    \label{lem:opposite-walk}
    Let $G$ be a bidirected graph and $G'$ the corresponding double covering graph.
    Let $W'$ be a walk in $G'$.
    Then its opposite $\hat{W}'$ is also a walk in $G'$.
\end{lemma}
\begin{proof}
    Forming the opposite of $W'$ is the same as applying \Cref{lem:connectivity} to get the walk $W$ that corresponds to $W'$ in $G$, then reversing $W$ and then applying \Cref{lem:connectivity} again to get the directed walk $\hat{W}'$ in the double covering graph that corresponds to the reversed $W'$.
\end{proof}

\section{Ultrabubbles}
\label{s:ultrabubbles}

\begin{figure}
    \centering
    \includegraphics[scale=1]{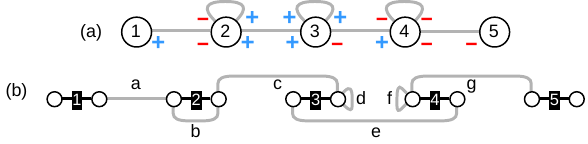}
    \caption{The biedged representation of a bidirected graph.
    \textbf{(a)} Bidirected graph.
    A walk alternates signs at each node.
    For example, $W := 1+, 2+, 2+, 3-, 4-, 4+, 3+, 3-, 4-, 5+$ is a walk.
    \textbf{(b)} Biedged graph.
    Nodes are subdivided into two and connected by a black edge.
    The left side represents incoming $-$ edges and the right side represents outgoing $+$ edges.
    Edges are represented as grey edges, and connected to the nodes according to their signs.
    A walk is an undirected walk that alternates between black and grey edges.
    $W$ is represented by the edge sequence $a, 2, b, 2, c, 3, e, 4, f, 4, e, 3, d, 3, e, 4, g, 5$.}
    \label{fig:biedged}
\end{figure}

\begin{figure}
    \centering
    \includegraphics[scale=1]{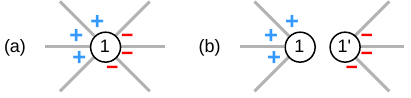}
    \caption{The splitting operation performed on signed node $1+$.}
    \label{fig:splitting}
\end{figure}

Ultrabubbles are defined in bidirected graphs.
While ultrabubbles were first defined by Paten~et~al.~\cite{paten2018ultrabubbles}, their definition works by first transforming the bidirected graph into a biedged graph.
See \Cref{fig:biedged} for an example of the biedged representation of a bidirected graph.
For simplicity, we use the equivalent definition in bidirected graphs given by Sena~et~al.~\cite{sena2025}.
For its definition, we first need to define the \emph{splitting} operation.
The splitting operation is applied to a signed node $v\alpha$ in a bidirected graph.
It inserts a new node $v'$ and replaces all signed nodes $v\hat{\alpha}$ in all edges with $v'\hat{\alpha}$.
See \Cref{fig:splitting} for an example of the splitting operation, and \Cref{fig:bubbles}~(c) for examples of the definition of ultrabubbles.

\begin{definition}[Ultrabubble and ultrabubble component~\cite{paten2018ultrabubbles,sena2025}]
    \label{def:ultrabubble-complex}
    Let $G$ be a bidirected graph. Let $\{u\alpha, v\beta\}$ be a pair of signed nodes with distinct $u, v \in V(G)$ and $\alpha, \beta \in \signs$. Then $\{u\alpha, v\beta\}$ is an \emph{ultrabubble} if:
    \begin{enumerate}[nosep]
        \item[(a)] \emph{separable:}
        the graph created by splitting $u\alpha$ and $v\beta$ contains a separate component $B \subseteq G$ containing $u$ and $v$ but not $u'$ and $v'$ (the new nodes inserted by the splitting operation).
        We call $B$ the \emph{ultrabubble component} of $\{u\alpha, v\beta\}$.
        
        \item[(b)] \emph{tipless:} no node in $V(B) \setminus \{u, v\}$ is a tip.
        
        \item[(c)] \emph{acyclic:} $B$ contains no cycloid.

        \item[(d)] \emph{minimal:}
        no signed node $w\gamma$ with node $w \in V(B) \setminus \{u,v\}$ is such that $\{u\alpha, w\gamma\}$ and $\{w\hat{\gamma}, v\beta\}$ are separable.
    \end{enumerate}
\end{definition}

\begin{figure}
    \centering
    \includegraphics[scale=1]{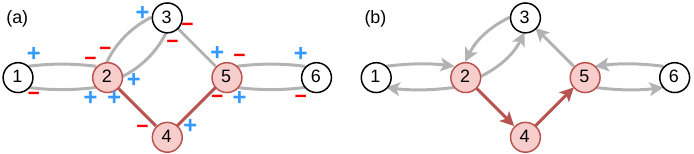}
    \caption{\textbf{(a)} By \Cref{def:ultrabubble}, the ultrabubble component of $\{2+, 5-\}$ (which is not an ultrabubble) contains the nodes $2$, $4$ and $5$ and edges $\{2+, 4-\}$ and $\{4+, 5-\}$ highlighted in red.
    All other nodes and edges are not on $2+$-$5+$-paths.
    \textbf{(b)} By \Cref{def:superbubble}, the weak superbubble component of $(2, 5)$ (which is not a weak superbubble) contains the nodes $2$, $4$ and $5$ and arcs $(2, 4)$ and $(4, 5)$ highlighted in red.
    All other nodes and arcs are not on $2$-$5$-paths.}
    \label{fig:components}
\end{figure}

We simplify this definition as follows.
See \Cref{fig:components}~(a) for an example of an ultrabubble component.

\begin{definition}[Ultrabubble and ultrabubble component (simplified)]
    \label{def:ultrabubble}
    Let $G$ be a bidirected graph. Let $\{u\alpha, v\beta\}$ be a pair of signed nodes with distinct $u, v \in V(G)$ and $\alpha, \beta \in \signs$.
    The \emph{ultrabubble component} $B$ of $\{u\alpha, v\beta\}$ is the graph consisting of $u$, $v$ and all nodes and edges on $u\alpha$-$v\hat{\beta}$-paths.
    Then $\{u\alpha, v\beta\}$ is an \emph{ultrabubble} if:
    \begin{enumerate}[nosep]
        \item[(a)] \emph{connected:} there is a $u\alpha$-$v\hat{\beta}$-path.
        
        \item[(b)] \emph{separable:} edges outside of $B$ are incident only to signed nodes outside of $B$, except that they may be incident to $u$ with sign $\hat{\alpha}$ or incident to $v$ with sign $\hat{\beta}$.
        
        \item[(c)] \emph{acyclic:} $B$ contains no cycloid.

        \item[(d)] \emph{minimal:}
        no signed node $w\gamma$ with node $w \in V(B) \setminus \{u,v\}$ is such that $\{u\alpha, w\gamma\}$ and $\{w\hat{\gamma}, v\beta\}$ satisfy the properties (a), (b) and (c) above.
    \end{enumerate}
\end{definition}

For proving the equivalence of the definitions, we first prove that: if $\{u\alpha, v\beta\}$ is an ultrabubble by the original definition, then the ultrabubble components are equivalent between the definitions.
Note that the simplified definition defines the ultrabubble component independently of $\{u\alpha, v\beta\}$ being an ultrabubble, and hence we can prove the following statement without proving the equivalence of the ultrabubble definitions first.

\begin{restatable}{lemma}{lemultrabubblecomponent}
    \label{lem:ultrabubble-component}
    Let $G$ be a bidirected graph and let $\{u\alpha, v\beta\}$ be an ultrabubble by \Cref{def:ultrabubble-complex}.
    Then the ultrabubble component $B$ of $\{u\alpha, v\beta\}$ by \Cref{def:ultrabubble-complex} is equivalent to the ultrabubble component $B'$ of $\{u\alpha, v\beta\}$ by \Cref{def:ultrabubble}.
\end{restatable}

\begin{restatable}[Equivalence of ultrabubble definitions]{theorem}{thmultrabubbledefinitionequivalence}
    \label{thm:ultrabubble-definition-equivalence}
    Let $G$ be a bidirected graph.
    Let $\{u\alpha, v\beta\}$ be a pair of signed nodes.
    Then $\{u\alpha, v\beta\}$ is an ultrabubble by \Cref{def:ultrabubble-complex} if and only if it is an ultrabubble by \Cref{def:ultrabubble}.
    If $\{u\alpha, v\beta\}$ is an ultrabubble, then the corresponding ultrabubble components are equivalent between the two definitions.
\end{restatable}

\section{Weak superbubbles}
\label{s:superbubbles}

Weak superbubbles are defined in directed graphs.
We are using Definition~2 of Gärtner~et~al.~\cite{gärtner2019direct}, which we reformulated into our terms.
Gärtner~et~al. define the set $U_{uv}$ of nodes reachable from node $u$ without passing through node $v$, as well as the set $U_{vu}^+$ of nodes that can reach node $v$ without passing through node $u$.
See \Cref{fig:weak-superbubbles}~(b) for an example of a (weak) superbubble.

\begin{definition}[Weak superbubble and weak superbubble component~\cite{gärtner2019direct}]
    \label{def:superbubble-original}
    Let $G$ be a directed graph and $U \subseteq V(G)$.
    Let $(u, v)$ be an ordered pair of distinct nodes of $U$ such that $U = U_{uv} = U_{vu}^+$.
    Then $(u, v)$ is a \emph{weak superbubble} with \emph{weak superbubble component} $U$ if:
    \begin{enumerate}[nosep]
        \item[(a)] every $x \in U$ is reachable from $u$,
        \item[(b)] every $x \in U$ reaches $v$,
        \item[(c)] if $x \in U$ and $y \in V(G) \setminus U$, then every $y$-$x$-path contains $u$,
        \item[(d)] if $x \in U$ and $y \in V(G) \setminus U$, then every $x$-$y$-path contains $v$,
        \item[(e)] if $(x, y)$ is an arc in $G$ with $x, y \in U$, then every $y$-$x$-path in $G$ contains both $v$ and $u$.
        \item[(f)] there is no node $v' \in U \setminus \{u, v\}$ such that $(u, v')$ satisfies properties (a)--(e) for some $U' \subset G$.
    \end{enumerate}
\end{definition}

Before we present our simplified definition of weak superbubbles, we observe that restricting $U = U_{uv} = U_{vu}^+$ introduces redundancy with conditions (a) through (d).
Hence we remove the extra restriction, and modify (a) and (b) to keep the definitions equivalent.

\begin{definition}[Weak superbubble and weak superbubble component (less redundant)]
    \label{def:superbubble-complex}
    Let $G$ be a directed graph and $U \subseteq V(G)$.
    Let $(u, v)$ be an ordered pair of distinct nodes of $U$.
    Then $(u, v)$ is a \emph{weak superbubble} with \emph{weak superbubble component} $U$ if:
    \begin{enumerate}[nosep]
        \item[(a)] every $x \in U$ is reachable from $u$ without passing through $v$,
        \item[(b)] every $x \in U$ reaches $v$ without passing through $u$,
        \item[(c)] if $x \in U$ and $y \in V(G) \setminus U$, then every $y$-$x$-path contains $u$,
        \item[(d)] if $x \in U$ and $y \in V(G) \setminus U$, then every $x$-$y$-path contains $v$,
        \item[(e)] if $(x, y)$ is an arc in $G$ with $x, y \in U$, then every $y$-$x$-path in $G$ contains both $v$ and $u$.
        \item[(f)] there is no node $v' \in U \setminus \{u, v\}$ such that $(u, v')$ satisfies properties (a)--(e) for some $U' \subset G$.
    \end{enumerate}
\end{definition}

\begin{restatable}[Equivalence of less redundant definition]{lemma}{lemequivalencelessredundant}
    Let $G$ be a directed graph.
    Let $(u, v)$ be an ordered pair of nodes.
    Then $(u, v)$ is a weak superbubble by \Cref{def:superbubble-original} if and only if it is a weak superbubble by \Cref{def:superbubble-complex}.
    If $(u, v)$ is a weak superbubble, $U$ is the weak superbubble component by \Cref{def:superbubble-original} and $U'$ is the weak superbubble component by \Cref{def:superbubble-complex}, then it holds that $U = U'$.
\end{restatable}

Now we simplify the less redundant definition of weak superbubbles further into fewer and simpler conditions that are analogous to the conditions for ultrabubbles.
See \Cref{fig:components}~(b) for an example of a weak superbubble component.

\begin{definition}[Weak superbubble and weak superbubble component (simplified)]
    \label{def:superbubble}
    Let $G$ be a directed graph. Let $(u, v)$ be an ordered pair of distinct nodes of $G$.
    The \emph{weak superbubble component} $B$ of $(u, v)$ is the graph consisting of $u$, $v$ and all nodes and arcs on $u$-$v$-paths.
    Then $(u, v)$ is a \emph{weak superbubble} if:
    \begin{enumerate}[nosep]
        \item[(a)] \emph{connected:} there is a $u$-$v$-path.
    
        \item[(b)] \emph{separable:} arcs outside of $B$ are incident only to nodes outside of $B$, except that they may be incoming to $u$ or outgoing from $v$.

        \item[(c)] \emph{acyclic:} $B$ contains no cycle.
        
        \item[(d)] \emph{minimal:} no node $w \in V(B) \setminus \{u, v\}$ is such that $(u, w)$ and $(w, v)$ satisfy the properties (a), (b) and (c) above.
    \end{enumerate}
\end{definition}

Note that in our simplified condition of minimality, we require that a weak superbubble cannot be split into two parts that are both weak superbubbles, while in the original condition \Cref{def:superbubble-complex}~(f), minimality only requires one part of the split to not be a weak superbubble.
We have chosen to formulate the minimality condition this way to make the similarity to ultrabubbles more obvious.
As we see in the proof of \Cref{thm:superbubble-definition-equivalence} below, due to symmetry, if we split a weak superbubble component that is connected, separable and acyclic, then either both parts of the split are connected, separable and acyclic or neither is.

In order to show the equivalence of the two weak superbubble definitions, we first re-derive the \emph{acyclicity} property from Gärtner~et~al.'s definition.
It was part of the original superbubble definition by Onodera~et~al.~\cite{onodera2013detecting}, but is not directly stated in Gärtner~et~al.'s definition of weak superbubbles.

\begin{restatable}[Acyclicity of weak superbubbles]{lemma}{lemacyclicity}
    \label{lem:acyclicity}
    Let $G$ be a directed graph and let $(u, v)$ be a weak superbubble by \Cref{def:superbubble-complex} in $G$ with weak superbubble component $U$.
    Then the graph $G[U] \setminus (v, u)$ contains no cycle.
\end{restatable}

Note that the weak superbubble component as in \Cref{def:superbubble-complex} is a set of nodes, but in our simplified \Cref{def:superbubble}, it is a graph.
This is to resemble the structure of the simplified ultrabubble definition, but also has the advantage that we do not need to worry about the arc $(v, u)$ for a weak superbubble $(u, v)$ as we do in \Cref{lem:acyclicity}.

To simplify our proof of equivalence, we first prove another implication.

\begin{restatable}{lemma}{lemsuperbubblecomponent}
    \label{lem:superbubble-component}
    Let $G$ be a directed graph and let $(u, v)$ be a weak superbubble by \Cref{def:superbubble-complex}.
    Let $U$ be its weak superbubble component by \Cref{def:superbubble-complex} and let $B$ be its weak superbubble component by \Cref{def:superbubble}.
    Then $G[U] \setminus (v, u) = B$.
\end{restatable}

\begin{restatable}[Equivalence of weak superbubble definitions]{theorem}{thmsuperbubbledefinitionequivalence}
    \label{thm:superbubble-definition-equivalence}
    Let $G$ be a directed graph.
    Let $(u, v)$ be an ordered pair of nodes.
    Then $(u, v)$ is a weak superbubble by \Cref{def:superbubble-complex} if and only if it is a weak superbubble by \Cref{def:superbubble}.
    If $(u, v)$ is a weak superbubble, $U$ is the weak superbubble component by \Cref{def:superbubble-complex} and $B$ is the weak superbubble component by \Cref{def:superbubble}, then $G[U] \setminus (v, u) = B$.
\end{restatable}

\section{Reduction}
\label{s:reduction}

Having harmonised the definitions of weak superbubbles and ultrabubbles, our reduction now becomes very simple due to the simplicity of the doubling operation.
See \Cref{fig:doubling}~(b) and~(c) for an example of the doubling operation.

The idea of our reduction is to prove an equivalence between a pair of signed nodes $U := \{u\alpha, v\beta\}$ being an ultrabubble and a corresponding pair $S := (u\alpha, v\hat{\beta})$ in the double covering graph being a weak superbubble.
While our simplified and harmonised definitions make this equivalence rather intuitive, the main challenge in our proof is that a cycloid in the bidirected graph does not have an immediate analogue in the double covering graph when the cycloid is not a cycle.
Hence, a path $P'$ in the double covering graph may translate into a cycloid $P$ in the bidirected graph instead of a path.
In this case, we informally refer to $P'$ as a cycloid.

To overcome this challenge, before proving the correctness of our reduction, we first exclude cycloids.
We show that if $S$ is a weak superbubble then the weak superbubble component $B'$ of $S$ contains no cycloid, and vice versa, that if $U$ is an ultrabubble then $B'$ contains no cycloid.

To show that $B'$ contains no cycloid if $S$ is a weak superbubble, we first show that $B'$ contains no cut nodes.
This means that any single node inside $B'$ can be avoided when traversing $B'$.
We prove this by contradiction, by showing that if there was a cut node $w$ inside $B'$, then $S$ could be split at $w$, such that the resulting two weak superbubble components decompose $B'$ while overlapping exactly on $w$.
This clean decomposition almost directly implies that $w$ violates the minimality of $S$.

\begin{restatable}[Weak superbubbles contain no cut nodes]{lemma}{lemcutnodes}
    \label{lem:cut-nodes}
    Let $G'$ be a directed graph.
    Let $S := (u, v)$ be a weak superbubble and $B$ its weak superbubble component.
    Then for all nodes $w \in V(B) \setminus \{u, v\}$, it holds that there is a $u$-$v$-path in $B$ that avoids $w$.
\end{restatable}

Next, we can show that $B'$ contains no cycloid.
Note that our statement in the lemma is stricter, forbidding only cycloids on $u\alpha$-$v\hat{\beta}$-paths.
But since $B'$ is acyclic, any cycloid can be extended into a $u\alpha$-$v\hat{\beta}$-path, and hence $B'$ contains no cycloid.
Therefore we informally use the more general statement to keep our explanations simpler.
To show that $B'$ contains no cycloid, we make a contradiction starting by assuming that there would be a $u\alpha$-$v\hat{\beta}$-path $P'$ that contains a cycloid.
Using its opposite path $\hat{P}'$, we show that then also $u\hat{\alpha}$ and $v\beta$ are in $B'$ because $S$ is separable.
Then we have enough opposite nodes in $B'$ such that we can ensure that $B'$ contains a cycle, which contradicts $S$ being a weak superbubble.

\begin{restatable}[Weak superbubbles contain no cycloid]{lemma}{lemnocycloid}
    \label{lem:no-cycloid}
    Let $G$ be a bidirected graph and $G'$ the corresponding double covering graph.
    Let $(u\alpha, v\hat{\beta})$ with $u \neq v$ be a weak superbubble in $G'$.
    Then there is no $u\alpha$-$v\hat{\beta}$-path that contains a pair of opposite nodes $w\gamma$ and $w\hat{\gamma}$.
\end{restatable}

Now we have excluded cycloids from $B'$ which makes the reverse direction of our equivalence proof below a lot simpler.
In the forward direction we assume that $U := \{u\alpha, v\beta\}$ is an ultrabubble and hence acyclic.
This implies that $P$ is in the ultrabubble component $B$ of $U$, and hence must be a path.

\begin{restatable}[Ultrabubbles imply cycloid-free weak superbubbles]{lemma}{lempathequivalence}
    \label{lem:path-equivalence}
    Let $G$ be a bidirected graph and $G'$ the corresponding double covering graph.
    Let $\{u\alpha, v\beta\}$ be an ultrabubble in $G$.
    Let $P'$ be a $u\alpha$-$v\hat{\beta}$-path in $G'$.
    Let $P$ be the walk corresponding to $P'$ by \Cref{lem:connectivity}.
    Then $P$ is a path.
\end{restatable}

Now we can prove the equivalence.
In the forward direction, connectedness follows straight away.
For separability and acyclicity, we use \Cref{lem:path-equivalence} to show that $B$ and $B'$ are equivalent restricted to the properties required in the conditions, and hence the conditions for $U$ imply the conditions for $S$.
In the reverse direction, we make use of \Cref{lem:no-cycloid} which states that $S$ contains no cycloid.
This gives us connectedness because the path connecting $S$ is also a path in the bidirected graph.
For separability and acyclicity we again show that $B$ and $B'$ are equivalent restricted to the properties required in the conditions, and hence the conditions for $S$ imply the conditions for $U$.
For minimality in both directions we make use of connectedness, separability and acyclicity from the opposite direction of the proof to say that if there was a node $w$ disproving minimality on one side, then it would also disprove minimality on the other side.

\begin{restatable}[Equivalence between ultrabubbles and weak superbubbles]{theorem}{thmequivalence}
    \label{thm:equivalence}
    Let $G$ be a bidirected graph and $G'$ the corresponding double covering graph.
    Then $\{u\alpha, v\beta\}$ is an ultrabubble in $G$ if and only if $u \neq v$ and $(u\alpha, v\hat{\beta})$ is a weak superbubble in $G'$.
\end{restatable}

In order to construct our reduction from \Cref{thm:equivalence}, note that each ultrabubble $\{u\alpha, v\beta\}$ is equivalent to two weak superbubbles, because \Cref{thm:equivalence} can also be applied to the reverse $\{v\beta, u\alpha\}$.
This implies the following.

\begin{observation}[Opposite weak superbubbles]
    \label{obs:opposite-weak-superbubbles}
    Let $G$ be a bidirected graph and $G'$ the corresponding double covering graph.
    Then $(u\alpha, v\hat{\beta})$ is a weak superbubble in $G'$ if and only if $(v\beta, u\hat{\alpha})$ is a weak superbubble in $G'$.
\end{observation}

Now we can construct a linear-time reduction, and to get a linear-time algorithm for enumerating ultrabubbles, we need a linear-time algorithm for enumerating weak superbubbles.
We use the algorithm of Gärtner~et~al.~\cite{gärtner-revisited}, specifically their Corollary~7 which can be restated in our terms as follows.

\begin{lemma}[Linear time identification \cite{gärtner-revisited}]
    \label{lem:linear-superbubbles}
    The weak superbubbles in a directed graph $G = (V, E)$ can be identified in $O(|V| + |E|)$ time.
\end{lemma}

\begin{theorem}
    In a bidirected graph $G$, all ultrabubbles can be enumerated in time $O(n + m)$.
\end{theorem}
\begin{proof}
    We double the graph in linear time by~\Cref{obs:doubling}.
    Then we enumerate all weak superbubbles in linear time by~\Cref{lem:linear-superbubbles}.
    We filter weak superbubbles of the form $(v\alpha, v\hat{\alpha})$ (which may occur e.g.~due to bidirected self loops).
    Then we translate the superbubbles into ultrabubbles according to \Cref{thm:equivalence}.
    Now we have up to two copies of each ultrabubble.
    Then we filter out weak superbubbles $(u\alpha, v\beta)$ where $u > v$ (if there is no natural way to order the nodes, then we assign them unique identifiers from $1$ to $n$ and use that as the order).
    Resulting, we have one weak superbubble per ultrabubble, which we then translate according to \Cref{thm:equivalence}.
    This all works in linear time.
\end{proof}

\section{Conclusions}

We have shown that ultrabubbles can be enumerated in linear time via doubling the bidirected graph and then enumerating weak superbubbles on the resulting directed graph.
This motivates that doubling a bidirected graph is not just a simplification technique for simple cases like shortest paths, but also for less obvious applications like the computation of ultrabubbles.
In the future, it would be interesting to investigate to what extent for example flow-like algorithms can be applied in bidirected graphs via graph doubling.

\bibliography{bibliography}

\appendix

\section{Omitted proofs}
\label{apx:proofs}

\lemultrabubblecomponent*
\begin{proof}
    It holds that $u, v \in B$ and $u, v \in B'$.

    $(\subseteq)$
    Let $x \in V(B) \setminus \{u, v\}$.
    Since $B$ is tipless, $x$ has at least one incident edge with positive and at least one incident edge with negative sign.
    And since the only tips in $B$ are $u$ and $v$ and $B$ is acyclic, any maximal walk from $x+$ must end in $u\hat{\alpha}$ or $v\hat{\beta}$.
    For the same reasons, any maximal walk from $x-$ must end in $u\hat{\alpha}$ or $v\hat{\beta}$.
    Since $B$ is acyclic, if $x+$ reaches $u\hat{\alpha}$ then $x-$ must reach $v\hat{\beta}$, and vice versa (if they both reach the same, then there would be a cycloid).
    Hence, $x$ is on a $u\alpha$-$v\hat{\beta}$-walk $W$.
    And if $W$ repeats any node (with the same or differing sides), then $B$ would contain a cycloid.
    Therefore, $W$ is a path.
    Hence, $x$ is on a $u\alpha$-$v\hat{\beta}$-path, and therefore $x \in B'$.

    Let $\{a\gamma, b\delta\} \in E(B)$.
    Since all nodes in $B$ are on $u\alpha$-$v\hat{\beta}$-paths, also the endpoints of $\{a\gamma, b\delta\}$ are on such paths.
    If there is one such path $P_1$ that contains $a\gamma$ and one such path $P_2$ that contains $b\hat{\delta}$ (we allow $P_1 = P_2$), then we can take $P_1$ until $a\gamma$ and append $P_2$ starting from $b\hat{\delta}$ to create a $u\alpha$-$v\hat{\beta}$-path that contains $\{a\gamma, b\delta\}$.
    Symmetrically, if there is a $u\alpha$-$v\hat{\beta}$-path that contains $a\hat{\gamma}$ and an $u\alpha$-$v\hat{\beta}$-path that contains $b\delta$, then there is a $u\alpha$-$v\hat{\beta}$-path that contains $\{a\gamma, b\delta\}$.
    On the other hand, if there are two $u\alpha$-$v\hat{\beta}$-paths that contain $a\gamma$ and $b\delta$, or if there are two $u\alpha$-$v\hat{\beta}$-paths that contain $a\hat{\gamma}$ and $b\hat{\delta}$, then these form a cycloid starting in $u\alpha$ and containing $\{a\gamma, b\delta\}$, which contradicts $B$ being acyclic.
    Hence, $\{a\gamma, b\delta\}$ is on a $u\alpha$-$v\hat{\beta}$-path, and therefore $\{a\gamma, b\delta\} \in B'$.

    $(\supseteq)$
    Let $x \in B'$ be an edge or a node such that $x \neq u$ and $x \neq v$.
    Then by definition, $x$ is on a $u\alpha$-$v\hat{\beta}$-path.
    And by definition of bidirected paths, such a path cannot contain $u\hat{\alpha}$ or $v\beta$.
    Thus, because $\{u\alpha, v\beta\}$ is separable, the path cannot leave $B$.
    Therefore, $x \in B$.
\end{proof}

\thmultrabubbledefinitionequivalence*
\begin{proof}
    $(\Rightarrow)$
    Let $\{u\alpha, v\beta\}$ be an ultrabubble by \Cref{def:ultrabubble-complex} and let $B$ be its ultrabubble component.
    Let $B'$ be its ultrabubble component by \Cref{def:ultrabubble}.
    By \Cref{lem:ultrabubble-component}, it holds that $B = B'$.

    \begin{enumerate}
        \item[(a)]
        Since $\{u\alpha, v\beta\}$ is separable by \Cref{def:ultrabubble-complex} it holds that $B$ is connected.
        And since $B = B'$, it also holds that $B'$ is connected.
        Hence there is a $u\alpha$-$v\hat{\beta}$-path and therefore $\{u\alpha, v\beta\}$ is connected by \Cref{def:ultrabubble}.
    
        \item[(b)]
        Since $\{u\alpha, v\beta\}$ is separable by \Cref{def:ultrabubble-complex} and $B = B'$, it holds that signed nodes $w\gamma$ are only contained in edges outside of $B$ if $w \notin B$ or if $w\gamma = u\hat{\alpha}$ or $w\gamma = v\hat{\beta}$.
        Therefore, it holds that $\{u\alpha, v\beta\}$ is separable by \Cref{def:ultrabubble}.
    
        \item[(c)]
        Further, because $B = B'$ and $B$ is acyclic by \Cref{def:ultrabubble-complex}, it follows that $B'$ is acyclic by \Cref{def:ultrabubble}.

        \item[(d)]
        Finally, $\{u\alpha, v\beta\}$ is minimal by \Cref{def:ultrabubble-complex}, i.e.~no signed node $w\gamma$ with node $w \in V(B) \setminus \{u, v\}$ exists such that $\{u\alpha, w\gamma\}$ and $\{w\hat{\gamma}, v\beta\}$ are separable by \Cref{def:ultrabubble-complex}.

        We split $u\alpha$, $w\gamma$ and $v\beta$ and call $C$ the component that contains $u$ and $w$ and $D$ the component that contains $w'$ and $v$.
        Since $\{u\alpha, v\beta\}$ is separable by \Cref{def:ultrabubble-complex}, but at least one of $\{u\alpha, w\gamma\}$ and $\{w\hat{\gamma}, v\beta\}$ is not separable, it holds that $C$ is connected to $D$ after splitting, and hence $C = D$.
        
        Further, since $u \in C$ and $v \in D$, there is an \emph{undirected} $u\alpha$-$v\hat{\beta}$-path $P$ in $C = D$.
        And since $w$ was split to create $C = D$, no undirected subpath of $P$ can form a (bidirected) path that passes through $w$.
        We memorize this path and use it again below.

        Next, we define $C'$ to be the ultrabubble component of $\{u\alpha, w\gamma\}$ by \Cref{def:ultrabubble} and we define $D'$ to be the ultrabubble component of $\{w\hat{\gamma}, v\beta\}$ by \Cref{def:ultrabubble}.
        Assume for a contradiction that $v \in C'$.
        Then there would be a $u\alpha$-$w\gamma$-path that contains $v$, so there is a $v\beta$-$w\gamma$-path $X$ or a $v\hat{\beta}$-$w\gamma$-path $\hat{X}$ that does not contain $u$.
        Because $B = B'$, it holds that $w\gamma$ is on a $u\alpha$-$v\hat{\beta}$-path.
        Hence, there is a $w\gamma$-$v\hat{\beta}$-path $Y$ that does not contain $u$.
        Then, either $XY$ or $\hat{X}Y$ forms a cycloid $Z$ through $v$ that does not contain $u$.
        Because $B$ is separable by \Cref{def:ultrabubble-complex} and does not contain $u$, it holds that $Z$ can enter and leave $B$ only through $v$, and hence $Z$ contains a sub-cycloid through $v$ inside $B$.
        This contradicts $B$ being acyclic by \Cref{def:ultrabubble-complex}, and hence $v \notin C'$.
        By symmetry, $u \notin D'$.

        Now we remember the undirected path $P$ from above.
        Since $u \in C'$ but $v \notin C'$, it holds that $P$ leaves $C'$ somewhere.
        Since $P$ starts in $u\alpha$, it holds that $P$ does not contain an edge with $u\hat{\alpha}$.
        And since $P$ contains no (bidirected) subpath that passes through $w$, it holds that $P$ leaves $C'$ through an edge that does not contain $w\hat{\gamma}$.
        Therefore $P$ leaves $C'$ somewhere else, and hence $\{u\alpha, w\gamma\}$ is not separable by \Cref{def:ultrabubble}.
        By symmetry, $\{w\hat{\gamma}, v\beta\}$ is not separable by \Cref{def:ultrabubble}.

        Hence, $\{u\alpha, v\beta\}$ is minimal by \Cref{def:ultrabubble}.
    \end{enumerate}

    Therefore, $\{u\alpha, v\beta\}$ is an ultrabubble by \Cref{def:ultrabubble} with the same ultrabubble component as by \Cref{def:ultrabubble-complex}.

    $(\Leftarrow)$
    Let $\{u\alpha, v\beta\}$ be an ultrabubble by \Cref{def:ultrabubble} and $B'$ its ultrabubble component.

    \begin{enumerate}
        \item[(a)]
        First, because $\{u\alpha, v\beta\}$ is separable by \Cref{def:ultrabubble}, it holds that edges outside of $B'$ contain only signed nodes outside of $B'$ or $u\hat{\alpha}$ or $v\hat{\beta}$.
        Further, by definition of $B'$, all edges in $B'$ contain only signed nodes inside $B'$.
        Hence, $\{u\alpha, v\beta\}$ is separable by \Cref{def:ultrabubble-complex} and $B' = B$, where $B$ is its ultrabubble component by \Cref{def:ultrabubble-complex}.

        \item[(b)]
        Next, by definition of $B'$, it contains no tips except for $u$ and $v$, and hence $B$ is tipless by \Cref{def:ultrabubble-complex}.

        \item[(c)]
        Further, because $B' = B$ and $B'$ is acyclic by \Cref{def:ultrabubble}, it follows that $B$ is acyclic by \Cref{def:ultrabubble-complex}.

        \item[(d)]
        Finally, $\{u\alpha, v\beta\}$ is minimal by \Cref{def:ultrabubble}, i.e.~no signed node $w\gamma$ with node $w \in V(B) \setminus \{u, v\}$ exists such that $\{u\alpha, w\gamma\}$ and $\{w\hat{\gamma}, v\beta\}$ fulfil conditions (a), (b) and (c) by \Cref{def:ultrabubble}.
        Assume for a contradiction that $\{u\alpha, v\beta\}$ could be subdivided into $\{u\alpha, w\gamma\}$ and $\{w\hat{\gamma}, v\beta\}$ such that $\{u\alpha, w\gamma\}$ and $\{w\hat{\gamma}, v\beta\}$ are separable by \Cref{def:ultrabubble-complex}.
        Then, because $\{u\alpha, v\beta\}$ is connected and acyclic by \Cref{def:ultrabubble} it also holds that $\{u\alpha, w\gamma\}$ and $\{w\hat{\gamma}, v\beta\}$ are connected and acyclic by \Cref{def:ultrabubble}.
        Further, as proven above, because $\{u\alpha, w\gamma\}$ and $\{w\hat{\gamma}, v\beta\}$ are separable by \Cref{def:ultrabubble-complex}, they are separable also by \Cref{def:ultrabubble}.
        But then, $\{u\alpha, w\gamma\}$ and $\{w\hat{\gamma}, v\beta\}$ would be connected, separable and acyclic, and hence $\{u\alpha, v\beta\}$ would not be minimal by \Cref{def:ultrabubble}, which contradicts the premise.
        Hence, at least one of $\{u\alpha, w\gamma\}$ and $\{w\hat{\gamma}, v\beta\}$ is not separable by \Cref{def:ultrabubble-complex}, and therefore $\{u\alpha, v\beta\}$ is minimal by \Cref{def:ultrabubble-complex}.
    \end{enumerate}

    Therefore, $\{u\alpha, v\beta\}$ is an ultrabubble by \Cref{def:ultrabubble-complex} with the same ultrabubble component as by \Cref{def:ultrabubble}.
\end{proof}

\lemequivalencelessredundant*
\begin{proof}
    $(\Rightarrow)$
    Let $(u, v)$ be a weak superbubble by \Cref{def:superbubble-original} with weak superbubble component $U$.
    In \Cref{def:superbubble-complex} we have removed the initial restrictions on $U$, and hence $U$ can be a weak superbubble component by \Cref{def:superbubble-complex} as well.
    Further, we modified conditions~(a) and~(b).
    Since by \Cref{def:superbubble-original}, $U$ is defined to contain only nodes in $U_{uv} \cap U_{vu}^+$, it holds that $(u, v)$ with $U$ fulfils these conditions in \Cref{def:superbubble-complex}.
    The remainder of the definition is unmodified.
    Hence, $(u, v)$ is a weak superbubble by \Cref{def:superbubble-complex} with weak superbubble component $U$.

    $(\Leftarrow)$
    Let $(u, v)$ be a weak superbubble by \Cref{def:superbubble-complex} with weak superbubble component $U$.
    Conditions~(a) and (b) in \Cref{def:superbubble-original} are weaker than those in \Cref{def:superbubble-complex}, and hence $(u, v)$ with $U$ fulfils these conditions of \Cref{def:superbubble-original}.
    \Cref{def:superbubble-original} restricts $U$ such that $U = U_{uv} = U_{vu}^+$.
    By \Cref{def:superbubble-complex}~(a) and~(b), it holds that $U \subseteq U_{uv}$ and $U \subseteq U_{vu}^+$.
    Further, if $U \neq U_{uv}$, then there would be a path out of $U$ that does not contain $v$, contradicting \Cref{def:superbubble-complex}~(d).
    Symmetrically, if $U \neq U_{vu}^+$, then there would be a path into $U$ that does not contain $u$, contradicting \Cref{def:superbubble-complex}~(c).
    Hence, it holds that $U = U_{uv} = U_{vu}^+$.
    The remainder of the definition is unmodified.
    Hence, $(u, v)$ is a weak superbubble by \Cref{def:superbubble-original} with weak superbubble component $U$.
\end{proof}

\lemacyclicity*
\begin{proof}
    Assume for a contradiction that there was a cycle $C$ in $G[U] \setminus (v, u)$.
    Without loss of generality, take $C$ as a closed path.
    We distinguish between $C$ containing $u$, $v$ or neither.

    \begin{itemize}
        \item
        If $C$ contains $u$, then let $(x, u)$ be the arc on $C$ that is incoming to $u$.
        Because of \Cref{def:superbubble-complex}~(a), there is a $u$-$x$-path $P$ that does not contain $v$.
        But $P$ contradicts \Cref{def:superbubble-complex}~(e).

        \item
        Symmetrically, if $C$ contains $v$, then let $(v, x)$ be the arc on $C$ that is outgoing from $v$.
        Because of \Cref{def:superbubble-complex}~(b), there is a $x$-$v$-path $P$ that does not contain $u$.
        But $P$ contradicts \Cref{def:superbubble-complex}~(e).

        \item
        If $C$ contains neither $u$ nor $v$, then it contradicts \Cref{def:superbubble-complex}~(e) for every arc in $C$.
    \end{itemize}
    
    Therefore, $C$ cannot exist.
\end{proof}

\lemsuperbubblecomponent*
\begin{proof}
    $(\subseteq)$
    Let $x \in U$.
    Then by \Cref{def:superbubble-complex}~(a), there is a $u$-$x$-path $P$ that does not contain $v$ and by~(b) there is an $x$-$v$-path $Q$ that does not contain $u$.
    Combined, $PQ$ is a $u$-$v$-walk containing $x$ that contains $u$ and $v$ only at its endpoints.
    Hence, by \Cref{def:superbubble-complex}~(c) and~(d), it holds that $PQ$ is in $U$.
    Further, if $PQ$ would repeat a node, then it would contain a cycle $C$ in $U$ that does not contain $u$ or $v$.
    But then $C$ would be in $G[U] \setminus (v, u)$ and hence by \Cref{lem:acyclicity} would contradict $(u, v)$ being a weak superbubble with weak superbubble component $U$.

    Let $e := (x, y) \in E(G[U] \setminus (v, u))$.
    Assume for a contradiction that $e \notin B$.
    Then every $u$-$v$-walk $W$ that contains $e$ would repeat a node $z$ and hence contain a cycle $C$.
    We can choose $W$ such that it does not repeat $u$ or $v$, and hence by \Cref{def:superbubble-complex}~(c) and~(d) we can choose $W$ and hence $C$ to be in $U$.
    Further, by choosing $W$ without repeating $u$ or $v$, also $C$ does not contain $u$ or $v$.
    But then $C$ would be in $G[U] \setminus (v, u)$ and hence by \Cref{lem:acyclicity} would contradict $(u, v)$ being a weak superbubble with weak superbubble component $U$.

    $(\supseteq)$
    Let $x \in B$ be a node or arc.
    Then there is a $u$-$v$-path $P$ inside $B$ that contains $x$.
    By \Cref{def:superbubble-complex}~(c) and~(d), because $u, v \in U$, it inductively holds that all nodes in $P$ are in $U$, including $x$ if it is a node.
    Because $P$ is a $u$-$v$-path, it does not contain the arc $(v, u)$, and hence $x$ is in $E(G[U] \setminus (v, u))$ if it is an arc.
\end{proof}

\thmsuperbubbledefinitionequivalence*
\begin{proof}
    $(\Rightarrow)$
    Let $(u, v)$ be a weak superbubble by \Cref{def:superbubble-complex} and let $U$ be its weak superbubble component.
    Let $B$ be its weak superbubble component by \Cref{def:superbubble}.
    By \Cref{lem:superbubble-component}, it holds that $G[U] \setminus (v, u) = B$.
    
    \begin{enumerate}
        \item[(a)]
        By definition, it holds that $u \in U$.
        Therefore, by \Cref{def:superbubble-complex}~(b) it holds that there is a $u$-$v$-path.
        Hence $(u, v)$ is connected by \Cref{def:superbubble}.
    
        \item[(b)]
        Let $(x, y) \in E(G) \setminus E(B)$ be an arc such that at least one of $x$ and $y$ is in $V(B)$.
        If $(x, y) = (v, u)$, then $(x, y)$ is outgoing from $v$ and incoming to $u$, and hence fulfils \Cref{def:superbubble}~(b).
        Otherwise, if $x \in V(B)$ and $y \in V(B)$, then $(x, y) \in E(G[U] \setminus (v, u)) = E(B)$, a contradiction.
        Hence, exactly one of $x$ and $y$ are in $V(B)$.
        But then, either $x = v$ or $y = u$, since otherwise $(x, y)$ would contradict $(u, v)$ being a weak superbubble by \Cref{def:superbubble-complex}~(c) and~(d).
        Hence, $(u, v)$ is separable by \Cref{def:superbubble}.
        
        \item[(c)]
        By \Cref{lem:acyclicity}, $B$ is acyclic by \Cref{def:superbubble}.

        \item[(d)]
        By \Cref{def:superbubble-complex}~(f), there is no node $v' \in U \setminus \{u, v\}$ such that $(u, v')$ is a weak superbubble by \Cref{def:superbubble-complex} without requiring property~(f).
        Since we have not used property~(f) in our proof above, it follows that there is no $v' \in B$ such that $(u, v')$ would be a weak superbubble by \Cref{def:superbubble} without requiring minimality.

        Further, note that in the reverse graph $\overleftarrow{G}$ of $G$ it holds that $(v, u)$ is a weak superbubble by \Cref{def:superbubble-complex}.
        Hence, by our argument above, it holds that there is no $u' \in B$ such that in $\overleftarrow{G}$ the pair $(v, u')$ would be a weak superbubble by \Cref{def:superbubble} without requiring minimality.
        Therefore, in $G$, $(u', v)$ would not be a weak superbubble by \Cref{def:superbubble} without requiring minimality.
        
        Therefore, $(u, v)$ is minimal by \Cref{def:superbubble}.
    \end{enumerate}
    
    Therefore, $(u, v)$ is a weak superbubble by \Cref{def:superbubble} with $G[U] \setminus (v, u) = B$.
    
    $(\Leftarrow)$
    Let $(u, v)$ be a weak superbubble by \Cref{def:superbubble} and let $B$ be its weak superbubble component.
    We choose $U = V(B)$ and show that with this choice of $U$ it holds that $(u, v)$ is a weak superbubble by \Cref{def:superbubble-complex} as well.

    \begin{enumerate}
        \item[(a)]
        By definition of $B$, every $x \in V(B)$ is on a $u$-$v$-path and by \Cref{def:superbubble}~(a) there exists at least one such path.
        Hence every $x \in U$ is reachable from $u$ without passing through $v$.
        
        \item[(b)]
        By definition of $B$, every $x \in V(B)$ is on a $u$-$v$-path and by \Cref{def:superbubble}~(a) there exists at least one such path.
        Hence every $x \in U$ reaches $v$ without passing through $u$.

        \item[(c)]
        Because $(u, v)$ is separable, the only way to enter $U$ is via $u$.
        Hence, for every $x \in U$ and $y \in V(G) \setminus U$ it holds that every $y$-$x$-path contains $u$.

        \item[(d)]
        Because $(u, v)$ is separable, the only way to leave $U$ is via $v$.
        Hence, for every $x \in U$ and $y \in V(G) \setminus U$ it holds that every $x$-$y$-path contains $v$.

        \item[(e)]
        Let $(x, y)$ be an arc in $G$ with $x, y \in U$.
        Let $P$ be a $y$-$x$-path.
        If $(x, y) = (v, u)$, then $P$ contains both $u$ and $v$.
        Otherwise, $(x, y)$ is in $B$ because $(u, v)$ is separable.
        Then, since $(u, v)$ is acyclic, it holds that $P$ contains both $u$ and $v$.

        \item[(f)]
        Because $(u, v)$ is minimal, no node in $B$ except for $v$ forms a pair with $u$ that forms a weak superbubble by \Cref{def:superbubble} without requiring property~(d).
        Since we have not used property~(d) in our proof above and $U = V(B)$, it follows that there is no node $v' \in U \setminus \{u, v\}$ such that $(u, v')$ is a weak superbubble by \Cref{def:superbubble-complex} without requiring property~(f).
    \end{enumerate}
    
    Therefore, $(u, v)$ is a weak superbubble by \Cref{def:superbubble-complex} with $G[U] \setminus (v, u) = B$.
\end{proof}

\lemcutnodes*
\begin{proof}
    Assume for a contradiction that there was a node $w \in V(B) \setminus \{u, v\}$ such that all $u$-$v$-paths contain $w$.
    Let $C$ be the weak superbubble component of $(u, w)$ and let $D$ be the weak superbubble component of $(w, v)$.
    We show that $C$ and $D$ overlap exactly in $w$, and that the union of $C$ and $D$ is exactly $B$.
    \begin{itemize}
        \item
        First, consider that any node or arc $x \in B$ is on some $u$-$v$-path $P$.
        Since all $u$-$v$-paths contain $w$, this implies that $x$ is also on some $u$-$w$-path or some $w$-$v$-path, and hence $x \in C \cup D$.

        \item
        Next, consider that any node or arc $x \in C$ is on some $u$-$w$-path $P$.
        Because $B$ is acyclic with $u$ being its only entrance and $v$ being its only exit, and $w \in B$, it holds that $P$ is in $B$.
        Hence, $x \in B$ and therefore $C \subseteq B$.
        Symmetrically, $D \subseteq B$.

        \item
        Finally, consider that there was a node or arc $x \in C \cap D$ where $x \neq w$.
        Then there would be both a $x$-$w$-path $P$ and a $w$-$x$-path $Q$ in $B$.
        But this contradicts $S$ being acyclic.
    \end{itemize}
    
    Therefore, $C$ and $D$ overlap exactly in $w$, and the union of $C$ and $D$ is exactly $B$.
    Then, $(u, w)$ and $(w, v)$ are
    \begin{itemize}
        \item[(a)] connected because $(u, v)$ is connected,
        \item[(b)] separable because $(u, v)$ is separable and $C$ and $D$ overlap exactly in $w$, and the union of $C$ and $D$ is exactly $B$, and
        \item[(c)] acyclic because $(u, v)$ is acyclic.
    \end{itemize}

    This contradicts $(u, v)$ being minimal, and hence for all nodes $w \in V(B) \setminus \{u, v\}$, it holds that there is a path in $B$ that avoids $w$.
\end{proof}

\lemnocycloid*
\begin{proof}
    Let $S := (u\alpha, v\hat{\beta})$ and let $B'$ be the weak superbubble component of $S$.
    We assume for a contradiction that there was a $u\alpha$-$v\hat{\beta}$-path $P'$ in $B'$ that contains a pair of opposite nodes $w\gamma$ and $w\hat{\gamma}$.
    We distinguish multiple cases to show that then also $u\hat{\alpha}, v\beta \in B'$.

    \begin{itemize}
        \item
        If $u\hat{\alpha}, v\beta \notin P'$, we consider its opposite $\hat{P}'$ which by \Cref{lem:opposite-walk} is a $v\beta$-$u\hat{\alpha}$-walk in $G'$.
        Since $w\gamma, w\hat{\gamma} \in P'$ and $u\hat{\alpha}, v\beta \notin P'$, it holds that $w\gamma, w\hat{\gamma} \in \hat{P}'$ and $u\alpha, v\hat{\beta} \notin \hat{P}'$.
        Then, we can see that $u\hat{\alpha}, v\beta \in B'$ by induction over $\hat{P}'$:
        it holds that $w\gamma \in B'$; and since $u\alpha, v\hat{\beta} \notin \hat{P}'$ and $S$ is separable, it holds that the predecessor and successor of every node in $\hat{P}'$ is also in $B'$.
        Hence, $u\hat{\alpha}, v\beta \in B'$.

        \item
        If $u\hat{\alpha} \in P'$ but $v\beta \notin P'$, then it holds that $u\hat{\alpha} \in B'$.
        Therefore, by \Cref{lem:cut-nodes}, there is a $u\alpha$-$v\hat{\beta}$-path $Z'$ that does not contain $u\hat{\alpha}$.
        The opposite $\hat{Z}'$ of $Z'$ is a $v\beta$-$u\hat{\alpha}$-walk by \Cref{lem:opposite-walk}.
        And since $u\hat{\alpha} \notin Z'$, it holds that $u\alpha \notin \hat{Z}'$.
        Then, we can see that $v\beta \in B'$ by induction over $\hat{Z}'$ in reverse:
        it holds that $u\hat{\alpha} \in B'$; and since $u\alpha \notin \hat{Z}'$ and $S$ is separable, it holds that the predecessor of every node in $\hat{Z}'$ is also in $B'$.
        Therefore, it holds that $v\beta \in B'$, and hence $u\hat{\alpha}, v\beta \in B'$.

        \item
        If $v\beta \in P'$ but $u\hat{\alpha} \notin P'$ we use an argument symmetric to the case above.
        Since $v\beta \in P'$, it holds that $v\beta \in B'$.
        Therefore, by \Cref{lem:cut-nodes}, there is a $u\alpha$-$v\hat{\beta}$-path $Z'$ that does not contain $v\beta$.
        The opposite $\hat{Z}'$ of $Z'$ is a $v\beta$-$u\hat{\alpha}$-walk by \Cref{lem:opposite-walk}.
        And since $v\beta \notin Z'$, it holds that $v\hat{\beta} \notin \hat{Z}'$.
        Then, we can see that $u\hat{\alpha} \in B'$ by induction over $\hat{Z}'$:
        it holds that $v\beta \in B'$; and since $v\hat{\beta} \notin \hat{Z}'$ and $S$ is separable, it holds that the successor of every node in $\hat{Z}'$ is also in $B'$.
        Therefore, it holds that $u\hat{\alpha} \in B'$, and hence $u\hat{\alpha}, v\beta \in B'$.

        \item
        If $u\hat{\alpha}, v\beta \in P'$, then because $P'$ is in $B'$, it holds that $u\hat{\alpha}, v\beta \in B'$.
    \end{itemize}

    Now, since $u\hat{\alpha}, v\beta \in B'$, by definition of $B'$, it holds that there is a $u\alpha$-$v\beta$-path $Q'$ in $B'$ and a $u\hat{\alpha}$-$v\hat{\beta}$-path $R'$ in $B'$.
    \begin{itemize}
        \item
        If $R'$ contains $v\beta$, then $B'$ contains a $u\hat{\alpha}$-$v\beta$-path which forms a cycle with $\hat{P}'$ in $B'$ and contradicts $B'$ being acyclic and hence contradicts $S$ being a weak superbubble.

        \item
        If $R'$ does not contain $v\beta$, then its opposite $\hat{R}'$ does not contain $v\hat{\beta}$, and hence never leaves $B'$ because $S$ is separable.
        But then, $B'$ contains both a $u\alpha$-$v\beta$-path $Q'$ and a $v\beta$-$u\alpha$-walk $\hat{R}'$, and hence a cycle.
        This contradicts $B'$ being acyclic and hence contradicts $S$ being a weak superbubble.
    \end{itemize}

    Since all branches lead to a contradiction it holds that $P'$ does not exist.
\end{proof}

\lempathequivalence*
\begin{proof}
    Assume for a contradiction that $P$ would not be a path.
    Since $P'$ is a path, this implies that $P$ can at most contain a cycloid that is not a cycle.
    Let $B$ be the ultrabubble component of $U := \{u\alpha, v\beta\}$.
    \begin{itemize}
        \item
        If $P$ is in $B$, then $B$ contains a cycloid, which contradicts $U$ being an ultrabubble.

        \item
        If $P$ is not in $B$, then $P$ still starts in $B$ because $u\alpha \in V(B)$.
        But since $P$ is not in $B$ this implies that $P$ leaves $B$ at some point.
        Because $U$ is separable, this implies that $P$ contains $v\hat{\beta}$ as inner node.
        But $P$ already ends in $v\hat{\beta}$, so it contains $v\hat{\beta}$ twice, which implies that $P'$ contains $v\hat{\beta}$ twice, which contradicts $P'$ being a path.
    \end{itemize}
    All cases result in a contradiction and therefore $P$ is a path.
\end{proof}

\thmequivalence*
\begin{proof}
    $(\Rightarrow)$
    Let $U := \{u\alpha, v\beta\}$ be an ultrabubble in $G$ and let $B$ be its ultrabubble component.
    Let $S := (u\alpha, v\hat{\beta})$ be a pair of nodes in $G'$ and let $B'$ be its weak superbubble component.
    Because $U$ is an ultrabubble, $u \neq v$.

    \begin{enumerate}
        \item[(a)]
        Because $U$ is connected, it holds that there exists a $u\alpha$-$v\hat{\beta}$-path in $G$.
        Hence, by \Cref{lem:connectivity}, there exists a $u\alpha$-$v\hat{\beta}$-path in $G'$, and hence $S$ is connected.
    
        \item[(b)]
        Let $e'$ be an arc outside of $B'$ that is incident to a node $w \in V(B')$.
        Assume for a contradiction that $e'$ was not incoming to $u\alpha$ and not outgoing from $v\hat{\beta}$.
        Let $e$ be the edge corresponding to $e'$ in $G$.
        It holds that $e \notin E(B)$, since otherwise $e$ would be on a $u\alpha$-$v\hat{\beta}$-path in $G$ and hence by \Cref{lem:connectivity} $e'$ would be on a $u\alpha$-$v\hat{\beta}$-path in $G'$ which contradicts the premise that $e' \notin E(B')$.
        On the other hand, because $w \in V(B')$, by definition of $B'$ there is a $u\alpha$-$v\hat{\beta}$-path $P'$ in $G'$ that contains $w$.
        Let $P$ be the walk in $G$ that corresponds to $P'$ by \Cref{lem:connectivity}.
        Since $P'$ contains $w$ it holds that $P$ contains $w$.
        By \Cref{lem:path-equivalence}, it holds that $P$ is a path.
        Therefore, it holds that $w \in V(B)$.

        Now, since $w \in V(B)$ but $e \notin E(B)$, it holds that $e$ is an edge outside of $B$ that is incident to a node $w \in V(B)$.
        Since $e'$ is not incoming to $u\alpha$, it holds that $u\hat{\alpha} \notin e$, and since $e'$ is not outgoing from $v\hat{\beta}$, it holds that $v\hat{\beta} \notin e$.
        But then, $e$ contradicts $U$ being separable, which contradicts the premise that $U$ is an ultrabubble.
        Hence, $e'$ is either incoming to $u\alpha$ or outgoing from $v\hat{\beta}$.
        Therefore, $S$ is separable.
        
        \item[(c)]
        Assume for a contradiction that $B'$ would contain a cycle $X'$.
        Then each node and arc $x \in X'$ would be on some $u\alpha$-$v\hat{\beta}$-path $P'$ in $G'$.
        For each $x$ let $P$ be the walk in $G$ that corresponds to $P'$ by \Cref{lem:connectivity}.
        By \Cref{lem:path-equivalence}, it holds that $P$ is a path.
        Therefore, the walk $X$ in $G$ that corresponds to $X'$ by \Cref{lem:connectivity} would be a cycle in $B$, which contradicts $U$ being an ultrabubble.
        Hence, $S$ is acyclic.

        \item[(d)]
        Assume for a contradiction that there was a node $w\gamma \in V(B')$ such that $(u\alpha, w\gamma)$ and $(w\gamma, v\hat{\beta})$ would be weak superbubbles without requiring minimality.
        Then by (a), (b) and (c) from the opposite direction of this proof below, the signed node $w\gamma$ in $G$ would contradict $U$ being minimal.
        Hence, $S$ is minimal.
    \end{enumerate}
    
    Hence, $S$ is a weak superbubble in $G'$.

    $(\Leftarrow)$
    Let $S := (u\alpha, v\hat{\beta})$ be a weak superbubble in $G'$ with weak superbubble component $B'$ and $u \neq v$.
    Let $U := \{u\alpha, v\beta\}$ be a set of signed nodes in $G$ and let $B$ be its ultrabubble component.
    
    \begin{enumerate}
        \item[(a)]
        Because $S$ is connected, there is a $u\alpha$-$v\hat{\beta}$-path $P'$ in $B'$.
        Let $P$ be the corresponding walk in $G$ by \Cref{lem:connectivity}.
        By \Cref{lem:no-cycloid}, it holds that $P'$ contains no pair of opposite nodes $w\gamma$ and $w\hat{\gamma}$.
        Therefore, $P$ is a path, and hence $U$ is connected.
    
        \item[(b)]
        Let $e := \{w\gamma, x\delta\}$ be an edge outside of $B$ that is incident to a node $w \in B$.
        Assume for a contradiction that $w\gamma \neq u\hat{\alpha}$ and $w\gamma \neq v\hat{\beta}$.
        We distinguish three cases to show that $e_1 := (w\gamma, x\hat{\delta}) \in B'$ or $e_2 := (x\delta, w\hat{\gamma}) \in B'$.
        
        \begin{itemize}
            \item
            If $w = u$, then $w\gamma = u\alpha$.
            But then $e_1 \in B'$ because $S$ is separable.

            \item
            If $w = v$, then $w\gamma = v\beta$ and hence $w\hat{\gamma} = v\hat{\beta}$.
            But then $e_2 \in B'$ because $S$ is separable.

            \item
            If $u \neq w \neq v$, consider that by definition of $B$ it holds that $w$ is on a $u\alpha$-$v\hat{\beta}$-path in $G$.
            And by \Cref{lem:connectivity} this implies that $w\epsilon$ is on some $u\alpha$-$v\hat{\beta}$-path in $G'$ for some $\epsilon \in \signs$, and hence $w\epsilon \in B'$.
            But then, since $S$ is separable, it holds that $e_1 \in B'$ if $\epsilon = \gamma$, or that $e_2 \in B'$ if $\epsilon = \hat{\gamma}$.
        \end{itemize}

        Since $e_1 \in B'$ or $e_2 \in B'$, there is a $u\alpha$-$v\hat{\beta}$-path $P'$ in $G'$ that contains $e_1$ or $e_2$.
        Because of \Cref{lem:no-cycloid}, it holds that $P'$ contains no pair of opposite nodes.
        Therefore, the walk $P$ in $G$ that corresponds to $P'$ by \Cref{lem:connectivity} is a path.
        Since $P'$ contains $e_1$ or $e_2$, it holds that $P$ contains $e$.
        Hence, $e \in E(B)$, which contradicts the premise that $e \notin E(B)$.
        Resulting, $U$ is separable.
        
        \item[(c)]
        Assume for a contradiction that $B$ contains a $w\gamma$-$w\delta$-cycloid $X$.
        By definition of $B$, there is a $u\alpha$-$v\hat{\beta}$-path in $B$ that contains $w\gamma$ or $w\delta$.
        By \Cref{lem:connectivity}, this implies that there is a $u\alpha$-$v\hat{\beta}$-path $P'$ in $B'$ that contains $w\gamma$ or $w\delta$.
    
        Let $X'$ be the $w\gamma$-$w\delta$-walk in $G'$ that corresponds to $X$ by \Cref{lem:connectivity}.
        Assume for a contradiction that $X'$ would contain $v\hat{\beta}$ other than as its last node.
        Then $X$ and hence $B$ would contain an edge $e$ that contains the signed node $v\hat{\beta}$.
        But for such an edge to appear on a $u\alpha$-$v\hat{\beta}$-path, the node $v$ would need to be repeated, because to be a path ending in $v\hat{\beta}$, the path has to enter $v$ through an edge that contains the signed node $v\beta$.
        This repetition contradicts $e$ being on a $u\alpha$-$v\hat{\beta}$-path, and hence $e \notin B$.
        Therefore, $X'$ contains $v\hat{\beta}$ at most as its last node.
        By symmetry, $X'$ contains $u\alpha$ at most as its first node.
        With this, we can see that $X'$ is in $B'$ by induction:
        it holds that $w\gamma \in B'$ or $w\delta \in B'$; and since $u\alpha, v\hat{\beta} \notin X'$ and $S$ is separable, it holds that the predecessor and successor of every node in $X'$ is also in $B'$.
        Therefore, it holds that $X'$ is in $B'$.
    
        Since $S$ is a weak superbubble and hence acyclic, $X'$ being in $B'$ means that $X'$ is not a cycle and hence $X$ is not a cycle.
        Further, it means that $X'$ contains no cycle.
        Therefore, by definition of $B'$ we can extend $X'$ into a $u\alpha$-$v\hat{\beta}$-walk $Y'$ in $B'$.
        Because $S$ is acyclic, it holds that $Y'$ is a path.
        And because $w\gamma, w\hat{\gamma} \in Y'$, it holds that $Y'$ contradicts \Cref{lem:no-cycloid}.
        Hence, $U$ is acyclic.

        \item[(d)]
        Assume for a contradiction that there was a signed node $w\gamma$ with $w \in V(B)$ such that $\{u\alpha, w\hat{\gamma}\}$ and $\{w\gamma, v\beta\}$ would be ultrabubbles without requiring minimality.
        Then by (a), (b) and (c) from the opposite direction of this proof above, the node $w\gamma$ in $G'$ would contradict $S$ being minimal.
        Hence, $U$ is minimal.
    \end{enumerate}
    
    Therefore, $U$ is an ultrabubble in $G$.
\end{proof}

\end{document}